\documentclass{eptcs}

\providecommand{\event}{CL\&C 2010}

\usepackage{proof,amssymb,stmaryrd}
\usepackage[all]{xy}

\newcommand{\detach}[3]{\hspace*{2mm} \deduce{\line(#1,6){#2}}{\makebox[0cm]
{#3}} \hspace*{2mm}}

\newcommand{\Int}{\mathbf{Int}}
\newcommand{\BiInt}{\mathbf{BiInt}}
\newcommand{\Ktfour}{\mathbf{KtT4}}
\newcommand{\Sfour}{\mathbf{S4}}
\newcommand{\Sfive}{\mathbf{S5}}

\newcommand{\seq}{\vdash}
\newcommand{\lseq}[1]{\vdash_{#1}}

\newcommand{\sneg}{\neg}
\newcommand{\wneg}{\mathop{\backsim}}
\newcommand{\con}{\wedge}
\newcommand{\dis}{\vee}
\newcommand{\imp}{\mathbin{\supset}}
\newcommand{\sub}{\mathbin{\Yleft}}

\newcommand{\from}{\leftarrow}

\newcommand{\hyp}{\mathit{hyp}}
\newcommand{\cut}{\mathit{cut}}
\newcommand{\weak}{\mathit{weak}}
\newcommand{\contr}{\mathit{contr}}
\newcommand{\nodesplit}{\mathit{nodesplit}}
\newcommand{\nodemerge}{\mathit{nodemerge}}

\newcommand{\monot}{\mathit{monot}}
\newcommand{\nest}{\mathit{nest}}
\newcommand{\unnest}{\mathit{unnest}}

\newcommand{\sglt}[1]{\langle #1 \rangle}
\newcommand{\join}[1]{\mathbin{\oplus}_{#1}}

\newcommand{\nodes}[1]{|#1|}
\newcommand{\dom}[1]{|#1|}

\newcommand{\downset}[2]{#1 \mathbin{\downarrow} #2}
\newcommand{\upset}[2]{#1 \mathbin{\uparrow} #2}

\newcommand{\ntol}[2]{\llbracket\,  #1 \, \rrbracket_{#2}}

\newcommand{\IH}[1]{\textrm{IH on $#1$}}

\newcommand{\sequent}[3]{#2\vdash_{#1}#3}
\newcommand{\SKIP}[1]{}

\newcommand{\dsc}{\mathbf{LBiI}}
\newcommand{\nsc}{\mathbf{N}\textbf{-}\mathbf{LBiI}}
\newcommand{\lsc}{\mathbf{L}\textbf{-}\mathbf{LBiI}}

\newcommand{\unL}[1]{\|{#1}\|^L}
\newcommand{\unR}[1]{\|{#1}\|^R}
\newcommand{\unLf}[1]{|{#1}|^L}
\newcommand{\unRf}[1]{|{#1}|^R}

\newcommand{\lton}[2]{\langle\!\langle~#1~\rangle\!\rangle_{#2}}

\newtheorem{theorem}{Theorem}
\newtheorem{lemma}{Lemma}
\newtheorem{proposition}{Proposition}

\newenvironment{proof}{\noindent \textbf{Proof}}{\hfill $\Box$}

\begin{document}
\title{Relating Sequent Calculi for \\ Bi-intuitionistic Propositional Logic}

\def\titlerunning{Relating Sequent Calculi for Bi-intuitionistic Propositional Logic}

\author{Lu\'is Pinto
\institute{Centro de Matem\'atica, Universidade do Minho, \\
Campus de Gualtar, P-4710-057 Braga, Portugal}
\email{luis@math.uminho.pt}
\and Tarmo Uustalu
\institute{Institute of Cybernetics at Tallinn University of Technology, \\
Akadeemia tee 21, EE-12618 Tallinn, Estonia}
\email{tarmo@cs.ioc.ee}
}

\def\authorrunning{L.~Pinto \& T.~Uustalu}

\maketitle

\begin{abstract}
  Bi-intuitionistic logic is the conservative extension of
  intuitionistic logic with a connective dual to implication. It is
  sometimes presented as a symmetric constructive subsystem of
  classical logic.

  In this paper, we compare three sequent calculi for bi-intuitionistic
  propositional logic: (1) a basic standard-style sequent calculus
  that restricts the premises of implication-right and exclusion-left
  inferences to be single-conclusion resp.\ single-assumption and is
  incomplete without the cut rule, (2) the calculus with nested
  sequents by Gor\'{e} et al., where a
  complete class of cuts is encapsulated into special ``unnest'' rules
  and (3)~a cut-free labelled sequent calculus derived from the Kripke
  semantics of the logic. We show that these calculi can be translated
  into each other and discuss the ineliminable cuts of the
  standard-style sequent calculus.
\end{abstract}


\section{Introduction}
\label{sec:intro}

\emph{Classical} logic is a logic of \emph{dualities}, some of which
are made evident with the help of its sequent calculus formulations.
For example, the duality between conjunction and disjunction is
manifest in the traditional sequent calculus rules for these
connectives, or the process of proving a formula $A$, i.e., deriving
the sequent $\seq A$, is dual to the process of refuting $A$, i.e.,
deriving the sequent $A \vdash $. However, if one considers
implication as a ``first-class'' connective and wants to retain
dualities, its dual connective of \emph{exclusion} (or
pseudo-difference, or subtraction) should also have a similar status.
In classical logic, exclusion is definable from other connectives:
$A\sub B$ is equivalent to $A\con\neg B$.  Classical logic with
exclusion in sequent calculus format has been considered by Curien and
Herbelin and Crolard \cite{CurienHerbelin,Crolard}, in connection to
the study of dualities in functional computation, e.g. the duality of
values and continuations or the duality of call-by-value and
call-by-name evaluation strategies.

\emph{Bi-intuitionistic} logic (also known as Heyting-Brouwer logic
and subtractive logic) results from classical logic with exclusion
by taking the implication to be \emph{intuitionistic} and, in order
to preserve duality, making exclusion \emph{dual-intuitionistic}
(hence the word '\emph{bi}'-intuitionistic). Bi-intuitionistic logic
can also be seen as the union of intuitionistic logic (lacking
exclusion) with dual-intuitionistic logic (lacking implication). The
former has the disjunction property while the latter has the dual
conjunction property (if $A\con B$ is refutable, either $A$ is
refutable or $B$ is refutable). In the union, both of these
properties are lost, yet one cannot prove excluded middle (nor
refute contradiction for the dual-intuitionistic weak negation).

Bi-intuitionistic logic first got attention by C.~Rauszer 
\cite{Rauszer74semiboolean,Rauszer74,Rauszer77}, who studied its
algebraic and Kripke semantics, alongside with an Hilbert-style system
and a sequent calculus. A sequent calculus characterization of the
logic is easily obtained by extending the multiple-conclusion
sequent calculus for intuitionistic logic \`a la Dragalin with
exclusion rules dual to the implication rules. But contrary to what is
suggested, e.g., in \cite{Restall}, cut is not fully eliminable in this
calculus. Neither is cut eliminable in the sequent calculus of Rauszer
\cite{Rauszer74} (the proof in the paper is incorrect).

The lack of cut elimination in a sequent calculus is problematic for
proof search, since the subformula property is not guaranteed. In
order to overcome this issue, extended sequent calculi have been
considered. In particular, a calculus of nested sequents was proposed
by Gor\'e, Postniece and Tiu~\cite{GorePostnieceTiu} and a calculus
with labelled formulas was proposed by the
authors~\cite{PintoUustalu}.

In the presence of various sequent calculi for bi-intuitionistic
logic, a natural question to ask is how they relate to each other.
In this paper, we give translations between the standard-style
(=Dragalin-style), nested and labelled sequent calculi for
bi-intuitionistic propositional logic. By analysing such translations,
we are able to identify a complete class of cuts for
the formulation with standard sequents and we are able to read the
proofs resulting from proof search in the nested or labelled sequent
calculus back into the standard-style sequent calculus. Reed and
Pfenning \cite{ReedPfenning} have remarked that relating labelled
intuitionistic derivations with standard unlabelled ones is
``a surprisingly difficult question''.


\section{The syntax and semantics of $\BiInt$}
\label{sec:biint}

We start by defining the syntax and semantics of bi-intuitionistic
propositional logic ($\BiInt$).

The language of $\BiInt$ extends that of intuitionistic propositional
logic ($\Int$), by one connective, exclusion, thus the {\em formulas}
are given by the grammar:
\[
A,B:=p\mid \top\mid \bot\mid A\con B\mid A\dis B\mid A\imp B\mid A\sub B
\]
where $p$ ranges over a denumerable set of {\em propositional
  variables} which give us atoms; the formula $A \sub B$ is the
\emph{exclusion} of $B$ from $A$. We do not take negations as
primitive, but in addition to the intuitionistic (or strong) negation,
there is also a dual-intuitionistic (or weak) negation. The two
negations are definable by $\sneg A:= A\imp \bot$ and $\wneg A:=
\top\sub A$.

The semantics of $\BiInt$ is usually given \`a la Kripke, although one
can also proceed from an algebraic semantics (in terms of
Heyting-Brouwer algebras) and there are further alternatives. The
Kripke semantics is about truth relative to worlds in Kripke
structures that are the same as for $\Int$. A \emph{Kripke structure}
is a triple $K = (W, \leq, I)$ where $W$ is a non-empty set whose
elements we think of as \emph{worlds}, $\leq$ is a preorder
(reflexive-transitive binary relation) on $W$ (the \emph{accessibility
  relation}) and $I$---the \emph{interpretation}---is an assignment of
sets of propositional variables to the worlds, which is monotone
w.r.t.\ $\leq$, i.e., whenever $w \leq w'$, we have $I(w) \subseteq
I(w')$.

\emph{Truth} in Kripke structures is defined as for $\Int$, but covers
also exclusion, interpreted dually to implication as possibility in
the past:
\begin{itemize}
\item $w \models p$ iff $p \in I(w)$;
\item $w \models \top$ always;
$w \models \bot$ never;
\item
$w \models A \con B$ iff $w \models A$ and $w \models B$;
$w \models A \dis B$ iff $w \models A$ or $w \models B$;
\item $w \models A \imp B$ iff, for any $w' \geq w$, $w' \not\models A$ or $w' \models B$;
\item $w \models A \sub B$ iff, for some $w' \leq w$, $w'
  \models A$ and $w' \not\models B$.
\end{itemize}
A formula is called \emph{valid} if it is true in all worlds of all
structures.  It is easy to see that monotonicity extends from atoms to
all formulas thanks to the universal and existential semantics of
implication and exclusion.

It is important for this paper that instead of general Kripke
structures, one may equivalently work with \emph{Kripke trees}.
These are Kripke structures $(W, \leq, I)$ where $W$ is finite and
the preorder $\leq$ arises as the reflexive-transitive closure of
some binary relation $\to$ on $W$, subject to the condition that any
two worlds $w$, $w'$ are related by the
reflexive-transitive-symmetric closure of $\to$ in a unique way
($w'$ is reached from $w$ by exactly one path along ${\to} \cup
{\from}$).

It is also a basic observation that the G\"odel translation of
$\Int$ into the modal logic $\Sfour$ extends to a translation into
the future-past tense logic $\Ktfour$ (cf.~\cite{Lukowski}). As the
semantics of $\Ktfour$ does not enforce monotonicity of
interpretations, atoms must be translated as future necessities or
past possibilities (these are always monotone): $p^\# = \Box p$ (or
$\blacklozenge p$); $\top^\# = \top$; $\bot^\#=\bot$; $(A \con B)^\#
= A^\# \con B^\#$; $(A \dis B)^\# = A^\# \dis B^\#$; $(A \imp B)^\#
= \Box (A^\# \imp B^\#)$; $(A \sub B)^\# = \blacklozenge (A^\# \sub
B^\#)$.


\section{Three sequent calculi for $\BiInt$}

We will now recall three different sequent calculi for $\BiInt$ that
we will later compare to each other.

\subsection{Standard-style sequent calculus $\dsc$}
\label{sec:dsc}

A sequent calculus for $\BiInt$ is most easily obtained by extending
Dragalin's sequent calculus for $\Int$, as has been done by Restall
\cite{Restall} and Crolard \cite{Crolard}.  (Rauszer's
\cite{Rauszer74} original sequent calculus was different.) In
Dragalin's system, sequents are multiple-conclusion, but the $\imp R$
rule is constrained. The extension, which we will now show, imposes
a dual constraint on the $\sub L$ rule.

The \emph{sequents} of our calculus (henceforth referred to as the
standard-style calculus $\dsc$) are pairs $\Gamma \seq \Delta$ where
$\Gamma, \Delta$ (the \emph{antecedent} and \emph{succedent}) are
finite multisets of formulas (we omit braces and denote union by
comma as usual). The inference rules of $\dsc$ are displayed in
Fig.~\ref{fig:dsc}.

\begin{figure}

\[
\hspace*{-5mm}
\begin{array}{cc}
\multicolumn{1}{l}{\textbf{Initial rule and cut (necessary):}}\\[2ex]
\infer[\hyp]{\Gamma, A \seq A, \Delta}{} \qquad \infer[\cut]{\Gamma
\seq \Delta}{
  \Gamma \seq A, \Delta
  &
  \Gamma, A \seq \Delta
}
\\[2ex]
\multicolumn{1}{l}{\textbf{Structural rules:}}
\\[2ex]
\infer[\weak L]{\Gamma, A \seq \Delta}{
  \Gamma \seq \Delta
} \qquad \infer[\weak R]{\Gamma \seq A, \Delta}{
  \Gamma \seq \Delta
} \qquad \infer[\contr L]{\Gamma, A \seq \Delta}{
  \Gamma, A, A \seq \Delta
} \qquad \infer[\contr R]{\Gamma \seq A, \Delta}{
  \Gamma \seq A, A, \Delta
}
\\[2ex]
\multicolumn{1}{l}{\textbf{Logical rules:}}\\[2ex]
\infer[\top L]{\Gamma, \top \seq \Delta}{
  \Gamma \seq \Delta
} \qquad \infer[\top R]{\Gamma \seq \top, \Delta}{} \qquad
\infer[\con L]{\Gamma, A \con B \seq \Delta}{
  \Gamma, A, B \seq \Delta
} \qquad \infer[\con R]{\Gamma \seq A \con B, \Delta}{
  \Gamma \seq A, \Delta
  &
  \Gamma \seq B, \Delta
}
\\[2ex]
\infer[\bot L]{\Gamma, \bot \seq \Delta}{} \qquad \infer[\bot
R]{\Gamma \seq \bot, \Delta}{
  \Gamma \seq \Delta
} \qquad \infer[\dis L]{\Gamma, A \dis B \seq \Delta}{
  \Gamma, A \seq \Delta
  &
  \Gamma, B \seq \Delta
} \qquad \infer[\dis R]{\Gamma \seq A \dis B, \Delta}{
  \Gamma \seq A, B, \Delta
}
\\[2ex]
\infer[\imp L]{\Gamma, A \imp B \seq \Delta}{
  \Gamma, A \imp B \seq A, \Delta
  &
  \Gamma, B \seq \Delta
} \qquad \infer[\imp R]{\Gamma \seq A \imp B, \Delta}{
  \Gamma, A \seq B
}
\\[2ex]
\infer[\sub L]{\Gamma, A \sub B \seq \Delta}{
  A \seq B, \Delta
} \qquad \infer[\sub R]{\Gamma \seq A \sub B, \Delta}{
  \Gamma \seq A, \Delta
  &
  \Gamma, B \seq A \sub B, \Delta
}
\end{array}
\]
\caption{Inference rules of $\dsc$}
\label{fig:dsc}
\end{figure}

Note that $\Delta$ is missing in the premise of the $\imp R$ rule and
dually in the premise of $\sub L$ we do not have the context $\Gamma$.

Regarding structural rules, both in $\dsc$ and the other two sequent
calculi ($\nsc$ and $\lsc$) considered in this paper, we have chosen
to work with formulations oriented at root-first proof search, which
means that, as a general guideline, we want to have our inference
rules ``as invertible as possible''. We have weakening and contraction
built in to the other rules to the degree that $\dsc$ and $\lsc$ are
complete without explicit versions of them.  This requires of course
that the two-premise rules are context-sharing etc. But there are also
more specific consequences. In $\dsc$, we have duplications of the
implication and exclusion formulas in the premises of the $\imp L$ and
$\sub R$ rules.

$\dsc$ is sound and complete for the Kripke semantics of $\BiInt$ for
the following generalization of validity from formulas to sequents. A
sequent $\Gamma \seq \Delta$ is taken to be valid if, for any in
Kripke structure $(W, \leq, I)$ and any world $w$, we have that if all
formulas in $\Gamma$ are true in $w$, then so is some formula in
$\Delta$. This has been proved (for variants of $\dsc$), e.g., by
Restall~\cite{Restall} and Monteiro~\cite{Monteiro}.

However, $\dsc$ is incomplete without cut, as shown by Pinto and
Uustalu in 2003 (private email message from T.~Uustalu to R.~Gor\'e,
13 Sept.\ 2004, quoted in \cite{BuismanGoreTableaux07}).  It suffices
to consider the obviously valid sequent $p \seq q, r \imp ((p \sub q)
\con r)$.  The only possible last inference (other than weakening and
contraction, which are redundant) in a derivation could be
\[
\infer[\imp R]{p \seq q, r \imp ((p \sub q) \con r)}{
  \deduce[?]{p, r \seq (p \sub q) \con r}{
  }
}
\]
but the premise is invalid as the succedent formula $q$ has been
lost. With cut, the sequent can be proved as follows:
\[
\infer[cut]{p \seq q, r \imp ((p \sub q) \con r)}{
  \infer[\sub R]{p \seq q, p \sub q, \ldots}{
    \infer[\hyp]{p \seq q, p, \ldots}{
    }
    &
    \infer[\hyp]{p, q \seq q, p \sub q, \ldots}{
    }
  }
  &
  \infer[\imp R]{p, p \sub q \seq q, r \imp ((p \sub q) \con r)}{
    \infer[\con R]{p, p \sub q, r \seq (p \sub q) \con r}{
      \infer[\hyp]{p, p \sub q, r \seq p \sub q}{
      }
      &
      \infer[\hyp]{p, p \sub q, r \seq r}{
      }
    }
  }
}
\]
Notice that permutation of the cut on the exclusion $p \sub q$ up past
the $\imp R$ inference, for which the cut formula is a side formula,
is not possible. This is one type of cuts that cannot be eliminated;
there are altogether three such types \cite{Monteiro}. This situation
reminds of the naive standard-style sequent calculus for $\Sfive$
where the sequent $p \seq \Box \Diamond p$ cannot be proved without
cut, but can be proved by applying cut to the sequents $p \seq
\Diamond p$ and $\Diamond p \seq \Box \Diamond p$ that are provable
without cut.

In Sec.~\ref{sec:appl}, with the help of the translations
proposed in this paper, we identify a class of cuts complete for
$\dsc$.

\subsection{Nested sequent calculus $\nsc$}
\label{sec:nsc}

Next we introduce a calculus $\nsc$ of nested sequents, which is a
minor variation of the calculus $\textrm{LBiInt}_1$ of Gor\'{e} et
al.~\cite{GorePostnieceTiu}.\footnote{The main difference is that
  $\textrm{LBiInt}_1$ does not build weakening and contraction into
  other rules, as we do in some cases to have a direct match with corresponding rules in the other systems
  considered.}
$\nsc$ is an extension of $\dsc$ where the concept of contexts is
generalized so that, alongside formulas, they can also contain nested
sequents, manipulated by dedicated additional inference rules.

The \emph{sequents} of $\nsc$ (ranged over by $S$) are defined
simultaneously with \emph{contexts} (ranged over by $\Gamma, \Delta$)
by the following grammar:
\[
\begin{array}{rcl}
S &::=& \Gamma \seq \Delta\\
\Gamma, \Delta&::=& \emptyset \mid A,\Gamma \mid S,\Gamma
\end{array}
\]
where contexts, just as in $\dsc$, are quotiented down to multisets
(so identified up to permutations of the member formulas/nested
sequents). Just as commas in antecedents and succedents intuitively
correspond to conjunctions and disjunctions, nested turnstiles should
be understood as structural-level implications and
exclusions.\footnote{In \cite{GorePostnieceTiu}, a nested sequent in
  the antecedent (resp.\ succedent) of a parent sequent (a
  structural-level exclusion resp.\ implication) is written
  $\Gamma<\Delta$ (resp.\ $\Gamma>\Delta$).}

The inference rules of $\nsc$ are those of $\dsc$ in
Fig.~\ref{fig:dsc} (including the cut rule and the structural rules)
together with additional inference rules for introducing and
eliminating nested sequents. These additional rules appear in
Fig.~\ref{fig:nsc}. The $\nest L$ and $\nest R$ rules are structural
versions of $\sub L$ and $\imp R$. The $\unnest L/R$ rules are
\emph{elimination} rules for exclusions on the left and implications
on the right. It is fair to think of them as masqueraded versions of
certain rather specific types of cuts (we come to this in
Sec.~\ref{sec:appl}).%
\footnote{For the sake of simplicity of presentation, we have opted
  for a formulation of $\unnest L/R$ rules that does not incorporate
  formula contraction. For a version of $\nsc$ that is complete
  without the $\contr L/R$ rules, the $\unnest L/R$ rules must be
  stated differently.}

\begin{figure}[t]

\[
\hspace*{-5mm}
\begin{array}{ccccc}
\multicolumn{2}{l}{\textbf{Rules for nested sequents:}}\\[2ex]
\infer[\nest L]{\Gamma, (\Gamma_0 \seq \Delta_0) \seq
\Delta}{
  \Gamma_0 \seq \Delta_0, \Delta
} &\quad& \infer[\nest R]{\Gamma \seq (\Gamma_0\seq\Delta_0),
\Delta}{
  \Gamma, \Gamma_0 \seq \Delta_0
}
\\[2ex]
\infer[\unnest L]{\Gamma, \Gamma_0 \seq \Delta_0,\Delta}{
  \Gamma, (\Gamma_0 \seq \Delta_0) \seq \Delta
} && \infer[\unnest R]{\Gamma,\Gamma_0\seq\Delta_0, \Delta}{
  \Gamma \seq (\Gamma_0\seq\Delta_0), \Delta
}
\end{array}
\]
\caption{Inference rules of $\nsc$ for manipulating nested sequents;
differently from $\dsc$, cut is redundant}
\label{fig:nsc}
\end{figure}

Stating soundness and completeness of $\nsc$ requires defining what
it means for a nested sequent to be valid.  This is achieved via a
translation that ``flattens'' nested sequents into standard
sequents, reducing validity of nested sequents to that of standard
sequents. We give a formal definition of this translation of
sequents in Sec.~\ref{sec:nsc2dsc}, where we show that derivations
of $\nsc$ can be translated into $\dsc$.  Gor\'e et
al.~\cite{GorePostnieceTiu} established soundness of $\nsc$ wrt.\
this notion of validity directly, but showed completeness by an
embedding of Rauszer's sequent calculus \cite{Rauszer74}.

They also showed cut to be redundant in the strong sense of existence
of a cut-eliminating transformation of derivations
\cite{GorePostnieceTiu}. The example of the previous section is proved
in $\nsc$ without cut (but with $\unnest L$) as follows:
\[
\infer[\unnest L]{p \seq q, r \imp ((p \sub q) \con r)}{
  \infer[\imp R]{(p \seq q) \seq r \imp ((p \sub q) \con r)}{
    \infer[\con R]{(p \seq q), r \seq (p \sub q) \con r}{
      \infer[\nest L]{(p \seq q), r \seq p \sub q}{
        \infer[\sub R]{p \seq q, p \sub q}{
          \infer[\hyp]{p \seq q, p}{
          }
          &
          \infer[\hyp]{p, q \seq q, p \sub q}{
          }
        }
      }
      &
      \infer[\hyp]{(p \seq q), r \seq r}{
      }
    }
  }
}
\]

\subsection{Labelled sequent calculus $\lsc$}
\label{sec:lsc}

The third sequent calculus we consider in this paper is a labelled
sequent calculus $\lsc$, a variation on the calculus L of ours
\cite{PintoUustalu}.%
\footnote{In fact, $\lsc$ lies between the calculi L and L$^*$ of
  \cite{PintoUustalu}. Similarly to L$^*$, the sequent calculus here
  is a calculus of finite Kripke trees rather than general Kripke
  structures, so we reason in terms of adjacency $\to$ rather than the
  induced accessibility relation ${\leq} = {\to}^*$. Differently from
  L, there are no reflexivity and transitivity rules, the monotonicity
  rules propagate truth/falsity to adjacent labels only (but can, of
  course, be applied multiple times), and the $\imp L$ and $\sub R$
  rules analyse and duplicate the main formula locally.}
The design of $\lsc$ follows the method of S.~Negri~\cite{Negri} for
obtaining cut-free sequent calculi for normal modal logics defined by
frame conditions of a certain type.  Essentially, $\lsc$ is a
formalization of the first-order theory of the Kripke semantics of
$\BiInt$, using an explicit device of labels for worlds.

A \emph{sequent} of $\lsc$ is a triple $\Gamma \lseq{G} \Delta$ where
$G$ is a label tree and $\Gamma$ and $\Delta$ are labelled contexts.
More precisely, the \emph{label tree} $G = (N, E)$ is a directed graph
that has its set of nodes $N$ (called labels) nonempty and finite and
is an \emph{undirected} tree in the sense that any nodes $x$, $y$ are
in the relation $(E \cup E^{-1})^*$ in only one way, i.e., are
connected by a single path of forward and backward arcs.  We write
$\nodes{G}$ for $N$ and $xGy$ for $(x, y) \in E$. The \emph{labelled
  contexts} $\Gamma$ and $\Delta$ are multisets of labelled formulas
and these, in their turn, are pairs $x : A$ with $x$ a label drawn
from $\nodes{G}$ and $A$ a formula.

Any label tree can be built (generally in many ways) from
constructions $\sglt{x}$ for the tree with a single node $x$, $(x,
y)$ for the tree with two nodes $x, y$ and one arc $xGy$, and $G
\join x G'$ for the join of two trees at $x$. In the last
construction, we require (as a welldefinedness condition) that the
trees $G$ and $G'$ satisfy $\dom{G} \cap \dom{G'} = \{x\}$, which
guarantees that the joint graph really is a tree.  Conversely, any
directed graph built in terms of these constructions is necessarily
a tree.

Intuitively, we use label trees to represent Kripke trees and a
labelled formula is about truth at a particular world.

The inference rules of $\lsc$ are presented in Fig.~\ref{fig:lsc}.
Some of them have provisos, that we also write as rule premises. The
conditions $\downset{G}{y}$ resp.\ $\upset{G}{y}$ mean that there is
no $z$ such that $zGy$ resp.\ $yGz$. The wellformedness condition of
any $\join{}$ expression occurrence must also be read as a proviso.
In the rules $\imp R$ and $\sub L$, we have freshness conditions on
$y$. Note the presence of the monotonicity rules intuitively
accounting for propagation of truth (resp.\ falsity) to future
(resp.\ past) worlds. The $\nodesplit$ and $\nodemerge$ structural
rules are auxiliary and redundant to the degree of existence of
eliminating transformations (alongside $\weak$ and $\contr$); we
included them here, because they come handy in the translation of
Sec.~\ref{sec:nsc2lsc}. $\nodesplit L/R$ split a node into a pair of
nodes connected by an arc, so that no paths are lost. $\nodemerge
L/R$ merge two nodes connected by an arc.

\begin{figure}
\[
\begin{array}{c}
\multicolumn{1}{l}{\textbf{Initial rule and cut (redundant):}}\\[2ex]
\infer[\hyp]{\Gamma,x : A \lseq{G} x : A, \Delta}{} \qquad
\infer[\cut]{\Gamma \lseq{G} \Delta}{
  \Gamma \lseq{G} x : A, \Delta
  &
  \Gamma, x : A \lseq{G} \Delta
}
\\[2ex]
\multicolumn{1}{l}{\textbf{Structural rules:}}\\[2ex]
\infer[\weak L]{\Gamma, x : A \lseq{G} \Delta}{
  \Gamma \lseq{G} \Delta
} \qquad \infer[\weak R]{\Gamma \lseq{G} x : A, \Delta}{
  \Gamma \lseq{G} \Delta
}
\\[2ex]
\infer[\contr L]{\Gamma, x : A \lseq{G} \Delta}{
  \Gamma, x : A, x : A \lseq{G} \Delta
} \qquad \infer[\contr R]{\Gamma \lseq{G} x : A, \Delta}{
  \Gamma \lseq{G} x : A, x : A, \Delta
}
\\[2ex]
\infer[\nodesplit U]{\Gamma \lseq{G_0 \join y (y, x) \join x G} \Delta}{
      \downset{G}{x} & \Gamma \lseq{G_0 \join y G[y/x]} \Delta}
\qquad \infer[\nodesplit D]{\Gamma \lseq{G \join x (x, y) \join y
G_0} \Delta}{
      \upset{G}{x} & \Gamma \lseq{G[y/x] \join y G_0} \Delta}
\\[2ex]
\infer[\nodemerge D]{\Gamma[x/y] \lseq{G_0[x/y] \join x G}
\Delta[x/y]}
  {\Gamma \lseq{G_0 \join y (y,x) \join x G} \Delta}
\qquad \infer[\nodemerge U]{\Gamma[x/y] \lseq{G \join x G_0[x/y]}
\Delta[x/y]}
  {\Gamma \lseq{G \join x (x,y) \join y G_0} \Delta}
\\[2ex]
\multicolumn{1}{l}{\textbf{Monotonicity rules:}}\\[2ex]
\infer[\monot
L]{\sequent{G}{\Gamma,x:A}{\Delta}}{xGy&\sequent{G}{\Gamma,x:A,y:A}{\Delta}}
\qquad \infer[\monot
R]{\sequent{G}{\Gamma}{x:A,\Delta}}{yGx&\sequent{G}{\Gamma}{y:A,x:A,\Delta}}
\\[2ex]
\multicolumn{1}{l}{\textbf{Logical rules:}}\\[2ex]
\infer[\top L]{ \sequent{G}{\Gamma, x : \top}{\Delta}}{
  \sequent{G}{\Gamma}{\Delta}
} \qquad \infer[\top R]{\sequent{G}{\Gamma}{ x : \top, \Delta}}{ }
\\[2ex]
\infer[\con L] {\sequent{G}{\Gamma, x : A \con B}{ \Delta}}{
\sequent{G}{ \Gamma, x : A, x : B}{ \Delta }} \qquad \infer[\con
R]{\sequent{G}{\Gamma}{x : A \con B, \Delta}}{
  \sequent{G}{\Gamma}{x : A, \Delta}
  &
  \sequent{G}{\Gamma}{x : B, \Delta}
}
\\[2ex]
\infer[\bot L]{\sequent{G}{\Gamma, x: \bot}{\Delta}}{ } \qquad
\infer[\bot R]{\sequent{G}{\Gamma}{x: \bot, \Delta}}{
  \sequent{G}{\Gamma}{\Delta}
}
\\[2ex]
\infer[\dis L]{\sequent{G}{\Gamma, x : A \dis B}{\Delta}}{
  \sequent{G}{\Gamma, x : A}{ \Delta}
  &
  \sequent{G}{\Gamma, x : B}{ \Delta}
} \qquad \infer[\dis R]{\sequent{G}{\Gamma}{x : A \dis B, \Delta}}{
  \sequent{G}{\Gamma}{ x : A, x : B, \Delta}
}
\\[2ex]
\infer[\imp L]{\sequent{G}{\Gamma, x : A \imp B}{\Delta}}{
  \sequent{G}{\Gamma, x : A \imp B}{x : A, \Delta}
  &
  \sequent{G}{\Gamma, x: B }{  \Delta} }
\qquad \infer[\imp R]{\sequent{G}{\Gamma}{x : A \imp B, \Delta}}{
  \sequent{G \join x (x,y)}{\Gamma, y : A}{y : B, \Delta}
}
\\[2ex]
\infer[\sub L]{\sequent{G}{\Gamma, x : A \sub B}{\Delta}}{
  \sequent{(y,x) \join x G}{\Gamma, y: A}{y : B, \Delta}
} \qquad \infer[\sub R]{\sequent{G}{\Gamma}{x : A \sub B, \Delta}}{
  \sequent{G}{\Gamma}{ x: A, \Delta}
  &
  \sequent{G}{\Gamma, x:B }{x : A \sub B, \Delta}
}
\end{array}
\]
\caption{Inference rules of $\lsc$}\label{fig:lsc}
\end{figure}

A labelled sequent $\Gamma \lseq{G} \Delta$ is \emph{valid} if, for
any Kripke structure $(W, \leq, I)$ and function $v : \nodes{G} \to W$
such that $xGy$ implies $v(x) \leq v(y)$, we have that, if $v(x)
\models A$ for every $x : A$ in $\Gamma$, then $v(x) \models A$ for
some $x : A$ in $\Delta$. $\lsc$ is sound and complete wrt.\ this
notion of validity.

$\lsc$ is complete without cut (as we proved by a semantic argument
in \cite{PintoUustalu}) and one should also be able to give a
cut-eliminating transformation of derivations.

Our counterexample to cut elimination in $\dsc$ is proved in $\lsc$ as
follows:
\[
\infer[\imp R]{x : p \lseq{\sglt{x}} x : q, x : r \imp ((p \sub q) \con r)}{
  \infer[\con R]{x : p, y : r \lseq{(x,y)} x : q, y : (p \sub q) \con r}{
     \infer[\monot R]{x : p, \ldots \lseq{(x,y)} x : q, y : p \sub q}{
      \infer[\sub R]{x : p, \ldots \lseq{(x,y)} x : q, x : p \sub q, \ldots}{
        \infer[\hyp]{x : p, \ldots \lseq{(x,y)} x : p, \ldots}{
        }
        &
        \infer[\hyp]{x : q, \ldots \lseq{(x,y)} x : q, \ldots}{
        }
      }
    }
    &
    \infer[\hyp]{x : p, y : r \lseq{(x,y)} x : q, y : r}{
    }
  }
}
\]
\normalsize
Notice the downward information propagation by the $\monot R$ inference
to an already existing label.


\section{Translations}

In this section, we study syntactic embeddings between the three
calculi.

We present six translations in all possible directions according to
the following plan (the sixth translation we will only sketch).
\[
\xymatrix{
 & \dsc \ar@/_5pt/[dr]
    \ar@/_5pt/[dl]_{\textrm{Sec.~\ref{sec:nsc2dsc}}}
   &
\\
\nsc \ar@/_5pt/[ur]
   \ar@/_5pt/[rr]_{\textrm{Sec.~\ref{sec:nsc2lsc}}}
 & &
   \lsc \ar@/_5pt/[ll]_{\textrm{Sec.~\ref{sec:lsc2nsc}}}
    \ar@{.>}@/_5pt/[ul]_{\textrm{Sec.~\ref{sec:dsc2lsc}}}
}
\]

\subsection{From $\nsc$ to $\dsc$ and back}
\label{sec:nsc2dsc}

As sequents and rules of $\dsc$ are also sequents and rules of $\nsc$,
a derivation in $\dsc$ is also a derivation in $\nsc$. Note, however,
that a cut in $\dsc$ is rendered by a cut also in $\nsc$. This is an
issue and we will reconsider it in section Sec.~\ref{sec:appl}.

For now, we move on to the converse direction.

We define simultaneously two functions on nested contexts
$\unLf{(-)}$ and $\unRf{(-)}$ that produce formulas. They are meant
to be applied to antecedents and succedents of sequents.  We also
introduce two further functions $\unL{(-)}$ and $\unR{(-)}$, defined
in terms of $\unLf{(-)}$ and $\unRf{(-)}$, to produce standard
contexts instead of formulas. They are used to translate top-level
sequents and avoid unnecessary rewriting of commas as $\con$ or
$\dis$.

\[
\begin{array}{rclcrcl}
\unLf{\emptyset}&=&\top&\quad&\unRf{\emptyset}&=&\bot\\
\unLf{A,\Gamma}&=&A\con\unLf{\Gamma}&\quad&\unRf{A,\Gamma}&=&A\dis\unRf{\Gamma}\\
\unLf{(\Gamma_0\vdash\Delta_0),\Gamma}&=&(\unLf{\Gamma_0}\sub\unRf{\Delta_0})\con\unLf{\Gamma}&\quad&\unRf{(\Gamma_0\vdash\Delta_0),\Gamma}&=&(\unLf{\Gamma_0}\imp\unRf{\Delta_0})\dis\unRf{\Gamma}\\[2ex]
\unL{\emptyset}&=&\emptyset&\quad&\unR{\emptyset}&=&\emptyset\\
\unL{A,\Gamma}&=&A,\unL{\Gamma}&\quad&\unR{A,\Gamma}&=&A,\unR{\Gamma}\\
\unL{(\Gamma_0\vdash\Delta_0),\Gamma}&=&(\unLf{\Gamma_0}\sub\unRf{\Delta_0}),\unL{\Gamma}&\quad&\unR{(\Gamma_0\vdash\Delta_0),\Gamma}&=&(\unLf{\Gamma_0}\imp\unRf{\Delta_0}),\unR{\Gamma}\\
\end{array}
\]

\begin{theorem}
  If $\Gamma \seq \Delta$ is derivable in $\nsc$, then $\unL{\Gamma} \seq
  \unR\Delta$ is derivable in $\dsc$.
\end{theorem}
\begin{proof}
  The proof is by induction on the structure of the $\nsc$ derivation
  of $\Gamma \seq \Delta$. The cases corresponding to rules other than
  the nesting rules are immediate, since there is a directly matching
  rule in $\dsc$.

Case $\nest R$: The given derivation has the form
\[
\infer[\nest R]{\Gamma \seq  (\Gamma_0 \seq \Delta_0),  \Delta}{
  \infer*[\pi]{\Gamma, \Gamma_0 \seq \Delta_0} {
  }
}
\]
It can be mapped to
\[
  \infer[\imp R]{\unL\Gamma \seq \unLf{\Gamma_0} \imp \unRf{\Delta_0}, \unR\Delta}{
    \infer=[(\con L, \dis R)^*]{\unL\Gamma,  \unLf{\Gamma_0} \seq \unRf{\Delta_0}}{
      \infer*[\IH{\pi}]{\unL\Gamma,  \unL{\Gamma_0} \seq \unR{\Delta_0}}{
      }
    }
  }
\]

Case $\unnest L$: The given derivation is of the form
\[
\hspace*{-5mm}
\infer[\unnest L] {\Gamma, \Gamma_0 \seq \Delta_0,\Delta} {
  \infer*[\pi]{\Gamma, (\Gamma_0 \seq \Delta_0) \seq \Delta} {
  }
}
\]
and we can transform it to
\[
\hspace*{-16mm} 
\infer[cut] {\unL{\Gamma}, \unL{\Gamma_0} \seq
\unR{\Delta_0},\unR{\Delta}} {
  \infer[\sub R]{\unL{\Gamma}, \unL{\Gamma_0} \seq \unLf{\Gamma_0}\sub\unRf{\Delta_0}, \unR{\Delta_0}, \unR{\Delta}}
{
    \infer=[(\con R)^*]{\ldots, \unL{\Gamma_0}\seq \unLf{\Gamma_0}, \ldots}{
      \forall i.\, \infer[\hyp]{\ldots, \unL{\Gamma_0}\seq \unL{\Gamma_0}_i, \ldots}{
      }
    }
    &
    \infer=[(\dis L)^*]{\ldots, \unRf{\Delta_0} \seq \unLf{\Gamma_0}\sub\unRf{\Delta_0},\unR{\Delta_0}, \ldots}{
      \forall i.\, \infer[\hyp]{\ldots, \unR{\Delta_0}_i \seq \unLf{\Gamma_0}\sub\unRf{\Delta_0},\unR{\Delta_0}, \ldots}{
      }
    }
  }
  &
  \hspace{-2mm}
  \detach{1}{10}{
  \hspace{-5.2cm}
  \infer=[(\weak L/R)^*]{\unL{\Gamma}, \unL{\Gamma_0}, \unLf{\Gamma_0}\sub\unRf{\Delta_0} \seq \unR{\Delta_0}, \unR{\Delta}}{
    \infer*[\IH{\pi}]{\unL{\Gamma}, \unLf{\Gamma_0}\sub\unRf{\Delta_0} \seq \unR{\Delta}}{
    }
  }
  }
}
\]

\end{proof}

\subsection{From $\lsc$ to $\nsc$}
\label{sec:lsc2nsc}

The translation from $\lsc$ to $\nsc$ is more involved than those of
the previous section, but also more illuminating.

The translation of a labelled sequent into a nested sequent follows
the idea that we can view any label of the label tree as its root
(intuitively, the focus of attention) and produce a nesting
structure for a nested sequent by mimicking this rooted tree.

The translation of a labeled sequent wrt.\ a chosen label from its
label tree is defined by recursion on the rooted tree structure by
\[
\begin{array}{rcl}
\lton{\Gamma \lseq{\sglt{x}} \Delta}{x}
& = & \Gamma(x) \seq \Delta(x) \\[1.2ex]
\lton{\Gamma \lseq{G \join{x} (x,y) \join{y} G_0} \Delta}{x}
& = & \Lambda \seq (\Lambda_0 \seq \Pi_0), \Pi \\
\multicolumn{3}{r}{\quad \textrm{where~} \Lambda \seq \Pi =
\lton{\Gamma[G] \lseq{G} \Delta[G]}{x}
   \textrm{~and~} \Lambda_0 \seq \Pi_0 = \lton{\Gamma[G_0] \lseq{G_0} \Delta[G_0]}{y}} \\[1.2ex]
\lton{\Gamma \lseq{G_0 \join{y} (y,x) \join{x} G} \Delta}{x}
& = & \Lambda, (\Lambda_0 \seq \Pi_0) \seq \Pi \\
\multicolumn{3}{r}{\quad \textrm{where~} \Lambda \seq \Pi =
\lton{\Gamma[G] \lseq{G} \Delta[G]}{x}
   \textrm{~and~} \Lambda_0 \seq \Pi_0 = \lton{\Gamma[G_0] \lseq{G_0} \Delta[G_0]}{y}}
\end{array}
\]
where $\Gamma(x) = \{ A \mid x : A \in \Gamma \}$ and $\Gamma[G] = \{
x : A \mid x \in \nodes{G} \textrm{~and~} x : A \in \Gamma \}$.

Intuitively, the formulas labelled with $x$ in the given sequent are
kept where they are, whereas those with labels reachable through the
labels immediately below resp.\ above $x$ are arranged into nested
sequent members of the antecedent resp.\ succedent of the top-level
nested sequent produced.

\begin{lemma}[Readdressing] For any $z, x \in \nodes{G}$, if
  $\lton{\Gamma \lseq{G} \Delta}{z}$ is derivable in $\nsc$, then so is $\lton{\Gamma \lseq{G}
    \Delta}{x}$.
\end{lemma}

\begin{proof} By induction on the unique path along $G\cup G^{-1}$
  from $x$ to $z$. The base case $x = z$ is trivial.

  We consider one of the two symmetric step cases, namely the one
  where $xGy$. In this case we have $G = G' \join{x} (x,y) \join{y} G_0$,
  with the path from $y$ to $z$ lying in $G_0$.

The given derivation is
\[
\infer*[\pi]{\lton{\Gamma \lseq{G} \Delta}{z}}{
}
\]
The nested sequent $\lton{\Gamma \lseq{G} \Delta}{x}$ can be derived
by
\[
\infer[\unnest L]{\Lambda \seq (\Lambda_0 \seq \Pi_0), \Pi}{
  \infer[\nest R]{(\Lambda\seq \Pi)\seq(\Lambda_0 \seq \Pi_0)}{
     \infer*[\IH{\pi}]{\Lambda_0, (\Lambda \seq \Pi) \seq \Pi_0}{
     }
  }}
\]
where $\Lambda \seq \Pi = \lton{\Gamma[G'] \lseq{G'} \Delta[G']}{x}$
and $\Lambda_0 \seq \Pi_0 = \lton{\Gamma[G_0] \lseq{G_0}
\Delta[G_0]}{y}$, so that $\lton{\Gamma \lseq{G} \Delta}{x} =
\Lambda \seq (\Lambda_0 \seq \Pi_0), \Pi$ whereas $\lton{\Gamma
\lseq{G} \Delta}{y} = \Lambda_0, (\Lambda \seq \Pi) \seq \Pi_0$.
\end{proof}

\begin{theorem}
  If $\Gamma \lseq{G} \Delta$ is derivable in $\lsc$, then $\lton{\Gamma \lseq{G}
    \Delta}{x}$ is derivable in $\nsc$ for any $x\in \nodes{G}$.
\end{theorem}
\begin{proof}
  By induction on the derivation of $\Gamma \lseq{G} \Delta$ in
  $\lsc$. We show the prototypical cases.

Case $\monot L$: The given derivation is of the form
\[
\hspace*{-5mm}
\infer[\monot L]{\Gamma,x :A \lseq{G \join{x} (x,y) \join{y} G_0} \Delta}{
  \infer*[\pi]{\Gamma, x :A, y:A \lseq{G \join{x} (x,y) \join{y} G_0} \Delta}{
  }
}
\]
By readdressing, it suffices to prove $\lton{\Gamma, x : A \lseq{G
    \join{x} (x,y) \join{y} G_0} \Delta}{x}$.

We construct this derivation:
\[
\infer[\contr L]{\Lambda, A \seq (\Lambda_0 \seq \Pi_0), \Pi}{
  \infer[\unnest L]{\Lambda, A. A \seq (\Lambda_0 \seq \Pi_0), \Pi}{
    \infer[\nest R]{(\Lambda. A \seq \Pi), A \seq (\Lambda_0 \seq \Pi_0)}{
       \infer*[\IH{\pi,y}]{(\Lambda, A \seq \Pi), \Lambda_0, A \seq \Pi_0}{
       }
    }
  }
}
\]
Here, $\Lambda \seq \Pi = \lton{\Gamma[G] \lseq{G} \Delta[G]}{x}$
and $\Lambda_0 \seq \Pi_0 = \lton{\Gamma[G_0] \lseq{G_0}
\Delta[G_0]}{y}$, which gives us $\lton{\Gamma, x : A \lseq{G
\join{x} (x,y) \join{y} G_0}
  \Delta}{x} = \Lambda, A \seq (\Lambda_0 \seq \Pi_0), \Pi$ and
$\lton{\Gamma, x : A, y : A \lseq{G \join{x} (x,y) \join{y} G_0}
  \Delta}{y} = (\Lambda, A \seq \Pi), \Lambda_0, A \seq \Pi_0$.

Case $\imp R$: The given derivation is of the form
\[
\infer[\imp R]{\Gamma \lseq{G} x : A \imp B, \Delta}{
  \infer*[\pi]{\Gamma, y: A \lseq{G \join x (x,y)} y : B, \Delta}{
  }
}
\]

We prove $\lton{\Gamma \lseq{G} x : A \imp B, \Delta}{x}$, which we
know is enough by readdressing. The derivation is this:
\[
\infer[\unnest L]{\Lambda \seq A \imp B, \Pi}{
  \infer[\imp R]{(\Lambda \seq \Pi) \seq A \imp B}{
    \infer*[\IH{\pi, y}]{(\Lambda \seq \Pi), A \seq B}{
    }
  }
}
\]
Here, $\Lambda \seq \Pi = \lton{\Gamma \lseq{G} \Delta}{x}$, which
gives us $\lton{\Gamma \lseq{G} x : A \imp B,
  \Delta}{x} = \Lambda \seq A \imp B, \Pi$ and
$\lton{\Gamma, y : A \lseq{G \join{x} (x,y)}
  y : B, \Delta}{y} = (\Lambda \seq \Pi), A \seq B$.
\end{proof}

\subsection{From $\nsc$ to $\lsc$}
\label{sec:nsc2lsc}

The translation from $\nsc$ to $\lsc$ is intended as an inverse for
that of the previous section. On sequents, it is a true inverse
(translating a sequent from $\nsc$ to $\lsc$ and back, we arrive at
exactly the same sequent; starting with an $\lsc$ sequent, we get an
isomorphic label tree with the same root). On derivations, the
isomorphism should hold up to a suitable notion of equivalence of
derivations on both sides (i.e., in both $\nsc$ and $\lsc$). We will
not pursue this here. But we expect that the right notions of
equivalence would be best formulated and the isomorphism established
with the help of term calculi.

We define a translation of $\nsc$ sequents to $\lsc$ sequents, by
induction on the antecedent and succedent of the given nested
sequent, by the following function, which also takes a label $x$ as
an additional argument. The root of the nesting structure of the
given nested sequent (i.e., its top level) is sent to label $x$ in
the label tree of the labelled sequent.
\[
\begin{array}{rcl}
\ntol{\seq}{x} & = & \lseq{\sglt{x}} \\
\ntol{\seq A, \Delta}{x}
& = & \Lambda \lseq{G} x : A, \Pi \\
&   & \quad \textrm{where~} \Lambda \lseq{G} \Pi = \ntol{\seq \Delta}{x} \\
\ntol{\seq (\Gamma_0 \seq \Delta_0), \Delta}{x}
& = & \Lambda, \Lambda_0 \lseq{G \join x (x,y) \join y G_0} \Pi_0, \Pi \\
&   & \quad \textrm{where~} \Lambda \lseq{G} \Pi = \ntol{\seq \Delta}{x}
      \textrm{~and~} \Lambda_0 \lseq{G_0} \Pi_0 = \ntol{\Gamma_0 \seq \Delta_0}{y} \\
\ntol{\Gamma, A \seq \Delta}{x}
& = & \Lambda, x : A \lseq{G} \Pi \\
&   & \quad \textrm{where~} \Lambda \lseq{G} \Pi = \ntol{\Gamma \seq \Delta}{x} \\
\ntol{\Gamma, (\Gamma_0 \seq \Delta_0) \seq  \Delta}{x}
& = & \Lambda, \Lambda_0 \lseq{G_0 \join y (y,x) \join x G} \Pi_0, \Pi \\
&   & \quad \textrm{where~} \Lambda \lseq{G} \Pi = \ntol{\Gamma \seq \Delta}{x}
      \textrm{~and~} \Lambda_0 \lseq{G_0} \Pi_0 = \ntol{\Gamma_0 \seq \Delta_0}{y}
\end{array}
\]
Intuitively, any formula in the top-level sequent is labelled by $x$
and remains where it is. Any sequent in the antecedent resp.\
succedent of the top-level sequent leads to the creation of a new
label $y$ immediately below resp.\ above $x$. The translated
elements of its antecedent resp.\ succedent are placed in the
antecedent resp.\ succedent of the sequent in the making.

Note that we have given the mathematical definition by first recursing
on the antecedent and then the succedent. In fact, the order is
immaterial, one could just as well start with the antecedent or,
indeed, remove formulas/nested sequents from the antecedent and
succedent in turns, in any order. This commutativity is used
extensively in our translation of derivations.

\begin{theorem}
  If $\Gamma \seq \Delta$ is derivable in $\nsc$, then $\ntol{\Gamma \seq
    \Delta}{x}$ is derivable in $\lsc$ for any $x$.
\end{theorem}

\begin{proof}
  By induction on the given derivation. We look at the following
  cases.

Case $\nest R$: The given derivation is of the form
\[
\infer[\nest R]{\Gamma \seq (\Gamma_0\seq\Delta_0), \Delta}{
  \infer*[\pi]{\Gamma, \Gamma_0 \seq \Delta_0}{
  }
}
\]

We can produce this derivation of the translated sequent:
\[
\infer=[(\weak L/R)^*]{\Lambda_d, \Lambda_u, \Lambda_0 \lseq{G_d \join x G_u   \join x (x, y) \join y G_0} \Pi_0, \Pi_d, \Pi_u}{
  \infer=[(\nodesplit U/D)^*]{\Lambda_d, \Lambda_0 \lseq{G_d \join x G_u    \join x (x, y) \join y G_0} \Pi_0, \Pi_d}{
    \infer=[(\monot L)^*]{\Lambda_d, \Lambda_0 \lseq{G_d \join x (x, y) \join y G_0} \Pi_0, \Pi_d}{
      \infer=[(\weak L)^*]{\Lambda_d[x,y/x], \Lambda_0  \lseq{G_d \join x (x, y) \join y G_0} \Pi_0, \Pi_d}{
        \infer[\nodesplit D]{\Lambda_d[y/x], \Lambda_0 \lseq{G_d \join x (x, y) \join y G_0} \Pi_0, \Pi_d[y/x] }{
          \infer*[(\IH{\pi}){[y/x]}]{\Lambda_d[y/x], \Lambda_0 \lseq{G_d[y/x] \join y G_0} \Pi_0, \Pi_d[y/x]}{
          }
        }
      }
    }
  }
}
\]
where $\Lambda_d \lseq{G_d} \Pi_d = \ntol{\Gamma \seq }{x}$,
$\Lambda_u \lseq{G_u} \Pi_u = \ntol{\seq \Delta}{x}$ and $\Lambda_0
\lseq{G_0} \Pi_0 = \ntol{\Gamma_0 \seq \Delta_0}{y}$, and
$\Lambda_d[x,y/x]$ stands for the union of $\Lambda_d[y/x]$ with the
context formed by the $x$-labelled formulas of $\Lambda_d$. Notice
that $x \notin \dom{\Pi_d}$, which tells us that $\Pi_d[y/x] =
\Pi_d$. The side condition of the topmost application of $\nodesplit
D$ is met because $\upset{G_d}{x}$. Note also that particular cases
of $\nodesplit U/D$ allow the addition of new nodes to a label tree.

Case $\unnest L$: We are given a derivation in the form
\[
\infer[\unnest L]{\Gamma,\Gamma_0\seq\Delta_0, \Delta}{
  \infer*[\pi]{\Gamma, (\Gamma_0\seq\Delta_0) \seq \Delta}{
  }
}
\]

We make the derivation
\[
\infer[\nodemerge D]{\Lambda, \Lambda_0 \lseq{G_0 \join x G} \Pi_0, \Pi}{
  \infer*[\IH{\pi}]{\Lambda, \Lambda_0[y/x] \lseq{G_0[y/x] \join y (y,x) \join x G} \Pi_0[y/x], \Pi}{
  }
}
\]
where $\Lambda \lseq{G} \Pi = \ntol{\Gamma \seq \Delta}{x}$ and
$\Lambda_0 \lseq{G_0} \Pi_0 = \ntol{\Gamma_0 \seq \Delta_0}{x}$.

Case $\imp R$: The given derivation is of the form
\[
\infer[\imp R]{\Gamma \seq A \imp B, \Delta}{
  \infer*[\pi]{\Gamma, A \seq B}{
  }
}
\]

We transform it to
\[
\infer=[(\weak L/R)^*]{\Lambda_d, \Lambda_u \lseq{G_d \join{x} G_u} x : A \imp B, \Pi_d, \Pi_u}{
  \infer=[(\nodesplit U/D)^*]{\Lambda_d \lseq{G_d \join{x} G_u} x : A \imp B, \Pi_d}{
     \infer[\imp R]{\Lambda_d \lseq{G_d} x : A \imp B, \Pi_d}{
      \infer=[(\monot L)^*]{\Lambda_d, y : A \lseq{G_d \join x (x, y)} y : B, \Pi_d }{
        \infer=[(\weak L)^*]{\Lambda_d[x,y/x], y : A \lseq{G_d \join x (x, y)} y : B, \Pi_d }{
          \infer[\nodesplit D]{\Lambda_d[y/x], y : A \lseq{G_d \join x (x, y)} y : B, \Pi_d[y/x]}{
            \infer*[(\IH{\pi}){[y/x]}]{\Lambda_d[y/x], y : A \lseq{G_d[y/x]} y : B, \Pi_d[y/x]}{
            }
          }
        }
      }
    }
  }
}
\]
where $\Lambda_d \lseq{G_d} \Pi_d = \ntol{\Gamma \seq }{x}$ and
$\Lambda_u \lseq{G_u} \Pi_u = \ntol{\seq \Delta}{x}$.  Notice that
$x \notin \dom{\Pi_d}$, with the effect that $\Pi_d[y/x] = \Pi_d$.
The side condition on the topmost application of $\nodesplit D$ is
satisfied as $\upset{G_d}{x}$.
\end{proof}

\subsection{From $\dsc$ into $\lsc$ and back}
\label{sec:dsc2lsc}

The translation of $\dsc$ into $\lsc$ is not demanding.  Essentially,
it suffices to annotate the end sequent with the sole label of a
singleton label tree and follow the structure of the $\dsc$-derivation
bottom-up, introducing new labels at $\imp R$ and $\sub L$. But again
(like in the translation from $\dsc$ to $\nsc$), a cut in $\dsc$ is
rendered by a cut in $\lsc$, which is not so perfect, since we should
not need cut in $\lsc$ derivations.

When we wrote the translation, we did not think of it like this, but
it can be described as composition of the translations from $\dsc$ to
$\nsc$ and from $\nsc$ further on to $\lsc$. Since $\dsc$ sequents
yield no nesting in $\nsc$, the readdressing that is needed in
translating derivations has only to do with the $\imp R$ and $\sub L$
rules.

Given a standard context $\Gamma$, we write $x : \Gamma$ for the
labelled context obtained by labelling all formulas of $\Gamma$ with
$x$.

\begin{theorem}
  If $\Gamma \seq \Delta$ is derivable in $\dsc$, then $x : \Gamma \lseq{\sglt{x}}
  x : \Delta$ is derivable in $\lsc$.
\end{theorem}

\begin{proof}
  By induction on the derivation of $\Gamma \seq \Delta$ in $\dsc$.
  We show one case.

Case $\imp R$: The given derivation
\[
\infer[\imp R]{\Gamma \seq A \imp B, \Delta}{
  \infer*[\pi]{\Gamma, A \seq B}{
  }
}
\]
is matched with the derivation
\[
\infer=[(\weak R)^*]{x:\Gamma \lseq{\sglt{x}} x : A \imp B, x:\Delta}{
\infer[\imp R]{x:\Gamma \lseq{\sglt{x}} x : A \imp B}{
  \infer[\monot L]{x:\Gamma,y: A \lseq{(x,y)} y: B}{
    \infer[\weak L]{x:\Gamma,y:\Gamma,y: A \lseq{(x,y)} y: B} {
       \infer[\nodesplit D]{y:\Gamma,y: A \lseq{(x,y)} y: B} {
         \infer*[\textrm{IH on $\pi$}]{y:\Gamma,y: A \lseq{\sglt{y}} y: B}{
         }
       }
     }
   }
}
}
\]
\end{proof}

The translation from $\lsc$ to $\dsc$ is best found by as a compound
translation through $\nsc$. We omit the details here, but it is quite
instructive. In particular, it gives a kind of explanation of why it
is so difficult to translate labelled derivations into standard
derivations in the case of $\Int$. We learn that the natural way uses
exclusion, and this is not available in $\Int$.


\section{Applications of the translations}
\label{sec:appl}

By analysing the targets of the various translations, one can find
some immediate applications. Our analysis essentially focuses on how
much cuts are needed in the translations, thus finding complete
classes of cuts. A direct use of the translations, not explored
here, is as a means of mapping proofs found by known search
procedures for $\BiInt$, based on the nested calculus and on the
labelled calculus, back into standard-style sequent calculus.

\subsubsection*{Translation from $\nsc$ into $\dsc$}

In this translation, the cut rule of $\dsc$ is used only for the
translation of the cut rule of $\nsc$ and of the unnesting-rules.
Let us call \emph{unnest cuts} the cuts of $\dsc$ of one of the
following two forms:

\[
\begin{array}{c}
\infer[\mathit{unnestcut} L]{\Gamma,\Gamma_0 \seq\Delta_0,\Delta} {
  \Gamma,\Gamma_0\seq\bigwedge\Gamma_0\sub\bigvee\Delta_0,\Delta_0,\Delta
  &
  \Gamma,\Gamma_0,\bigwedge\Gamma_0\sub\bigvee\Delta_0 \seq \Delta_0,\Delta
} \\[2ex]
 \infer[\mathit{unnestcut} R]{\Gamma,\Gamma_0
\seq\Delta,\Delta_0} {
  \Gamma,\Gamma_0 \seq \bigwedge\Gamma_0\imp\bigvee\Delta_0,\Delta_0,\Delta
  &
  \Gamma,\Gamma_0, \bigwedge\Gamma_0\imp\bigvee\Delta_0 \seq\Delta_0,\Delta
}
 \end{array}
\]
Observe that these two special cases of cut are the ones used in the
translations of the $\unnest$ rules. As $\nsc$ is complete without
cut, we have that unnest cuts are complete for $\dsc$.

\begin{proposition}
  The sequent calculus obtained from $\dsc$ by restricting to unnest
  cuts is complete for $\BiInt$.
\end{proposition}

Now of course the first premise of $unnestcut L$ and the second
premise of $unnestcut R$ are derivable, so a more practical idea would
be to remove cuts altogether and instead make the rules

\[
\infer[\unnest L]{\Gamma,\Gamma_0 \seq\Delta_0,\Delta} {
  \Gamma,\Gamma_0, \bigwedge\Gamma_0\sub\bigvee\Delta_0 \seq
  \Delta_0,\Delta
} \qquad \infer[\unnest R]{\Gamma,\Gamma_0 \seq\Delta_0,\Delta} {
  \Gamma,\Gamma_0 \seq
  \bigwedge\Gamma_0\imp\bigvee\Delta_0,\Delta_0,\Delta
}
\]

\subsubsection*{From $\nsc$ to $\dsc$ and back and then there again}
The attempt at a direct cut elimination transformation for $\dsc$
fails because of three combinations of cuts with other rules
\cite{Monteiro}. They fall into two wider combinations:

\[
\begin{array}{c}
\infer[\cut]{\Gamma \seq C \imp D,  \Delta}{
  \Gamma \seq A, C \imp D, \Delta
  &
  \infer[\imp R]{\Gamma, A \seq C \imp D, \Delta}{
    \Gamma, A, C \seq D
  }
} \quad\quad \infer[\cut]{\Gamma, C \sub D \seq \Delta}{
  \infer[\sub L]{\Gamma, C \sub D \seq A, \Delta}{
    C \seq D, A, \Delta
  }
  &
  \Gamma, C \sub D ,A \seq \Delta
}
\end{array}
\]

The $\unnest L$ and $\unnest R$ rules give a possibility to permute
the cuts up past the $\imp R$ and $\sub L$ inferences in these two
configurations.

We show this for the $\imp R$ case:
\[
  \infer[\unnest L]{\Gamma \seq C \imp D, \Delta}{
    \infer[\imp R]{\Gamma, (\bigwedge \Gamma \sub \bigvee (\Delta,C \imp D)) \seq C \imp D,\Delta}{
      \infer[\cut]{\Gamma, (\bigwedge \Gamma \sub \bigvee (\Delta,C \imp D)), C \seq D}{
         \infer[\sub L]{\Gamma,\bigwedge \Gamma \sub \bigvee (\Delta,C \imp D), C \seq A, D}{
           \infer=[(\con L, \dis R)^*]{\bigwedge \Gamma \seq \bigvee (\Delta,C \imp D), A, D}{
             \infer[\weak R]{\Gamma \seq \Delta, C \imp D, A, D}{\Gamma \seq \Delta, C \imp D, A}
           }
         }
         &
         \infer[\weak L]{\Gamma,(\bigwedge \Gamma \sub \bigvee (\Delta,C \imp D)), C, A \seq
         D}{\Gamma, C, A \seq
         D}
      }
    }
  }
\]
The possibility of this permutation may give a transformation
replacing the cuts in a $\dsc$ derivation with $\unnest L$ and
$\unnest R$, but we have not checked if it is welldefined, i.e.,
terminates.

\subsubsection*{Translations between $\lsc$ into $\nsc$}
In these translations, cuts in the target are only needed for
translating cuts in the source. Thus, redundancy of cut in one
implies redundancy of cut in the other.


\section{Final remarks}

Our translations between standard-style, nested and labelled sequent
calculi for $\BiInt$ provide a framework for comparison of proof
transformations within each of these systems.  A basic question is
to understand the relationship between ways of performing
cut-elimination in the three systems: (i) for the nested system,
Gor\'e et al.~\cite{GorePostnieceTiu} have described a
cut-elimination procedure; (ii) for the labelled system, one should
be able to adapt Negri's procedure \cite{Negri}, which applies to a
wide range of normal modal logics; (iii) for the standard-style
system, cuts are not fully eliminable, but unnest cuts are a
complete and simple form of cuts.

A tool that should be helpful to perform comparison of
cut-elimination processes is term assignment. We are not aware of
term assignment done directly for systems based on nested sequents.
This kind of formalisms have been used mostly in connection with
proof search and modal logics, exploiting the subformula property
\cite{BrunnlerHabilitation}, but for $\BiInt$, besides the study of
shallow inference and nested sequents of Gor\'e et
al.~\cite{GorePostnieceTiu}, a study of deep inference and nested
sequents (though still with emphasis on proof search) is also
available \cite{Postniece-deep-inference}.  As to labelled systems,
Reed and Pfenning \cite{ReedPfenning} consider term assignment in
the context of labelled intuitionistic logic.  They work with
natural deduction and use control operators \emph{letcc} and
\emph{throw} to account for the multiple conclusions of the labelled
sequent calculus that they proceed from. A term assignment for
$\BiInt$ corresponding to Dragalin's style can be obtained from the
well-studied calculus $\overline{\lambda}{\mu\tilde\mu}$ with the
typing system $LK_{\mu\tilde\mu}$ of Curien-Herbelin
\cite{CurienHerbelin}.  From these systems, which are for classical
logic (or for classical logic with exclusion \cite{CurienHerbelin}),
we obtain the $\imp,\sub$-fragment of $\BiInt$ by imposing the usual
single formula restriction in the succedent (resp.\ antecedent) of
the rule corresponding to $\imp R$ (resp.\ $\sub L$), i.e.,
\[
\infer[\imp R] {\Gamma\vdash \lambda x.v:A\imp
B|\Delta}{\Gamma,x:A\vdash v:B|}\qquad \infer[\sub L]
{\Gamma|\beta\lambda.e:B\sub A\vdash\Delta}{|e:B\vdash
\beta:A,\Delta}
\]

The exclusion operation has been given computational meaning by,
e.g., Crolard \cite{CrolardJLC} and Ariola et
al.~\cite{Ariola-et-al}. Crolard considered multiple-conclusion
natural deduction systems both for classical logic with exclusion
and for $\BiInt$, the latter being obtained from the former by a
mechanism to keep track of dependencies between hypotheses and
conclusions and ensure the $\BiInt$-restrictions, arriving at a
\emph{safe $\lambda$-calculus}, where coroutines (a restricted form
of continuations) become first-class objects. Ariola et al.\
considered classical logic with exclusion in a natural deduction
system close to that of Crolard's, to provide a typing system for a
$\lambda$-calculus with delimited continuations. It would be
interesting to revisit these ideas in connection to the sequent
calculi studied in this paper and understand whether the sequent
calculus format (standard or extended) has anything new to offer. As
Crolard's mechanism to keep track of dependencies resembles our
labelled system for $\BiInt$, a specific goal would be to
investigate relationships between the two systems.


\paragraph{Acknowledgments} We are grateful to our anonymous referees
for their helpful comments.  We are also grateful to Linda Postniece
for discussions. We were supported by the Portuguese Foundation for
Science and Technology through Centro de Matem\'{a}tica da
Universidade do Minho and project RESCUE no.\ PTDC/EIA/65862/2006,
ERDF through the Estonian Centre of Excellence in Computer Science,
EXCS, and the Estonian Science Foundation under grant no.~6940.



\end{document}